\documentclass[11pt]{article}
\usepackage{amsmath,amsfonts,amssymb,amsthm}
\usepackage{fullpage}
\usepackage[ruled]{algorithm2e}

\usepackage{subfigure} \usepackage{epstopdf}
\usepackage{graphicx}
\usepackage{color}
\usepackage[sans]{dsfont}
\usepackage{mathtools}
\usepackage{hyperref}
\usepackage{cleveref}
\usepackage{tcolorbox}
\usepackage{multirow}
\usepackage{xspace}
\tcbuselibrary{skins,breakable}
\tcbset{enhanced jigsaw}

\newtheorem{theorem}{Theorem}[section]

\newtheorem{corollary}[theorem]{Corollary}
\newtheorem{lemma}[theorem]{Lemma}
\newtheorem{observation}[theorem]{Observation}
\newtheorem{proposition}[theorem]{Proposition}

\theoremstyle{definition}
\newtheorem{claim}[theorem]{Claim}

\newcommand{\set}[1]{\left\{ #1 \right\}}

\newcommand{\pset}{{\mathcal{P}}}

\newcommand{\dset}{{\mathcal{D}}}

\newcommand{\fset}{{\mathcal{F}}}

\newcommand{\eps}{{\varepsilon}}

\newcommand{\diam}{\textnormal{\textsf{diam}}}
\newcommand{\dist}{\textnormal{\textsf{dist}}}
\newcommand{\gir}{\textnormal{\textsf{girth}}}
\newcommand\vol{\mathsf{cost}}

\def\ceil#1{\left\lceil #1 \right\rceil}

\def\card#1{\left| #1 \right|}
\def\deg{\textrm{deg}}

\def\pr#1{\mathrm{Pr}\left[ #1 \right]}

\def\ex#1{{\mathbb{E}}\left[ #1 \right]}

\newcounter{note}

\newcommand{\zesn}{\mathsf{0EwSN}}
\newcommand{\ze}{$0$-$\textsf{Ext}$\xspace}

\begin{document}

\begin{titlepage}
	
	\title{Lower Bounds on $0$-Extension with Steiner Nodes}
	
	\author{Yu Chen\thanks{EPFL, Lausanne, Switzerland. Email: {\tt yu.chen@epfl.ch}. Supported by ERC Starting Grant 759471.} \and Zihan Tan\thanks{Rutgers University, NJ, USA. Email: {\tt zihantan1993@gmail.com}. Supported by a grant to DIMACS from the Simons Foundation (820931).}} 
	
	\maketitle

	\thispagestyle{empty}
	\begin{abstract}
		
In the \emph{$0$-Extension problem}, we are given an edge-weighted graph $G=(V,E,c)$, a set $T\subseteq V$ of its vertices called terminals, and a semi-metric $D$ over $T$, and the goal is to find an assignment $f$ of each non-terminal vertex to a terminal, minimizing the sum, over all edges $(u,v)\in E$, the product of the edge weight $c(u,v)$ and the distance $D(f(u),f(v))$ between the terminals that $u,v$ are mapped to.
Current best approximation algorithms on $0$-Extension are based on rounding a linear programming relaxation called the \emph{semi-metric LP relaxation}. The integrality gap of this LP, with best upper bound $O(\log |T|/\log\log |T|)$ and best lower bound $\Omega((\log |T|)^{2/3})$, has been shown to be closely related to the best quality of cut and flow vertex sparsifiers.

We study a variant of the $0$-Extension problem where Steiner vertices are allowed. Specifically, we focus on the integrality gap of the same semi-metric LP relaxation to this new problem. 
Following from previous work, this new integrality gap turns out to be closely related to the quality achievable by cut/flow vertex sparsifiers with Steiner nodes, a major open problem in graph compression. 
Our main result is that the new integrality gap stays superconstant $\Omega(\log\log |T|)$ even if we allow a super-linear $O(|T|\log^{1-\varepsilon}|T|)$ number of Steiner nodes.

	\end{abstract}
\end{titlepage}

\section{Introduction}

In the \emph{$0$-Extension} problem (\ze), we are given an undirected edge-weighted graph $G=(V,E,c)$, a set $T\subseteq V$ of its vertices called \emph{terminals}, and a metric $D$ on terminals, and the goal is to find a mapping $f: V\to T$ that maps each vertex to a terminal in $T$, such that each terminal is mapped to itself (i.e., $f(t)=t$ for all $t\in T$), and the sum $\sum_{(u,v)\in E}c{(u,v)}\cdot D(f(u),f(v))$ is minimized.

The \ze problem was first introduced by Karzanov \cite{karzanov1998minimum}. It is a generalization of the \emph{multi-way cut} problem (by setting $D(t,t')=1$ for all pairs $t,t'\in T$) \cite{dahlhaus1994complexity,cualinescu1998improved,freund2000lower,buchbinder2013simplex,angelidakis2017improved,buchbinder2017simplex,berczi2020improving}, and a special case of the \emph{metric labeling} problem \cite{kleinberg2002approximation,chekuri2004linear,archer2004approximate,karloff2006earthmover,chuzhoy2007hardness}.
C\u{a}linescu, Karloff and Rabani \cite{calinescu2005approximation} gave the first approximation algorithm for \ze, achieving a ratio of $O(\log |T|)$, by rounding the solution of a \emph{semi-metric LP relaxation} (\textsf{LP-Metric}), which is presented below.
\begin{eqnarray*}
	\mbox{(\textsf{LP-Metric})}\quad &  \text{minimize}\quad  \sum_{(u,v)\in E}c(u,v)\cdot \delta(u,v)\\
	& s.t. \quad (V,\delta) \text{ is a semi-metric space}\\
	& \delta(t,t')=D(t,t'), \quad\forall t,t'\in T
\end{eqnarray*}
Fakcharoenphol, Harrelson, Rao and
Talwar \cite{fakcharoenphol2003improved} later gave a modified rounding algorithm on the same LP, improving the ratio to $O(\log |T|/\log\log |T|)$, which is the current best-known approximation. On the other hand, this LP was shown to have integrality gap $\Omega(\sqrt{\log |T|})$ \cite{calinescu2005approximation}, and this was recently improved to $\Omega((\log |T|)^{2/3})$ by Schwartz and Tur \cite{schwartz2021metric}.
Another LP relaxation called \emph{earthmover distance relaxation} (\textsf{LP-EMD}) was considered by Chekuri, Khanna, Naor and Zosin \cite{chekuri2004linear} and utilized to obtain an $O(\log |T|)$-approximation of the metric labeling problem (and therefore also the \ze problem). It has been shown \cite{karloff2009earthmover} by Karloff, Khot, Mehta and Rabani that this LP relaxation has an integrality gap $\Omega(\sqrt{\log |T|})$. Manokaran, Naor, Raghavendra
and Schwartz \cite{manokaran2008sdp} showed that the integrality gap of this LP relaxation leads to the same hardness of approximation result, assuming the Unique Game Conjecture.

In addition to being an important problem on its own, the \ze problem and its two LP relaxations are also closely related to the construction of \emph{cut/flow vertex sparsifiers}, a central problem in the paradigm of graph compression.
Given a graph $G$ and a set $T\subseteq V(G)$ of terminals, a cut sparsifier of $G$ with respect to $T$ is a graph $H$ with $V(H)=T$, such that for every partition $(T_1,T_2)$ of $T$, the size of the minimum cut separating $T_1$ from $T_2$ in $G$ and the size of the minimum cut separating $T_1$ from $T_2$ in $H$, are within some small factor $q$, which is also called the \emph{quality} of the sparsifier\footnote{flow sparsifiers has a slightly more technical definition, which can be found in \cite{hagerup1998characterizing,leighton2010extensions,chuzhoy2012vertex,andoni2014towards}.}.
Moitra \cite{moitra2009approximation} first showed that every graph with $k$ terminals admits a cut sparsifier with quality bounded by the \emph{integrality gap} of its \textsf{LP-Metric} (hence $O(\log k/\log\log k)$). Later on, Leighton and Moitra \cite{leighton2010extensions}, and Makarychev and Makarychev \cite{makarychev2010metric} concurrently obtained the same results for flow sparsifiers, and then Charikar, Leighton, Li and Moitra \cite{charikar2010vertex} showed that the best flow sparsifiers can be computed by solving an LP similar to \textsf{LP-EMD}. On the lower bound side, it was shown after a line of work \cite{leighton2010extensions,charikar2010vertex,makarychev2010metric} that there exist graphs with $k$ terminals whose best flow sparsifier has quality $\tilde\Omega(\sqrt{\log k})$.

A major open question on vertex sparsifiers is: 

\vspace{-10pt}
\[\emph{\underline{Q1.} Can better quality sparsifiers be achieved by allowing a small number of Steiner vertices}? 
\]
\vspace{-10pt}

In other words, what if we no longer require that the sparsifier $H$ only contain terminals, but just require that $H$ contain all terminals and its size be bounded by some function $f$ on the number of terminals (for example, $f(k)=2k,k^2$ or even $2^k$)? 
Chuzhoy \cite{chuzhoy2012vertex} constructed $O(1)$-quality cut/flow sparsifiers with size dependent on the number of terminal-incident edges in $G$.
Andoni, Gupta and Krauthgamer \cite{andoni2014towards} showed the construction for $(1+\eps)$-quality flow sparsifiers for quasi-bipartite graphs. For general graphs, they constructed a sketch of size $f(k,\eps)$ that stores all feasible multicommodity flows up to a factor of $(1+\eps)$, raising the hope for a special type of $(1+\eps)$-quality flow sparsifier, called contraction-based flow sparsifiers, of size $f(k,\eps)$ for general graphs, which was recently invalidated by Chen and Tan \cite{chen20241+}, who showed that contraction-based flow sparsifiers whose size are bounded by any function $f(k)$ must have quality $1+\Omega(1)$. But it is still possible for such flow sparsifiers with constant quality and finite size to exist. Prior to this work, Krauthgamer and Mosenzon \cite{krauthgamer2023exact} showed that there exist $6$-terminal graphs $G$ whose quality-$1$ flow sparsifiers must have an arbitrarily large size.

Given the concrete connection between the \ze problem and  cut/flow sparsifiers, it is natural to ask a similar question for \ze: 

\vspace{-10pt}
\[\emph{\underline{Q2.} Can better approximation of \textnormal{\ze} be achieved by allowing a small number of Steiner vertices?}\]
\vspace{-10pt}

In this paper, we formulate and study the following variant of the \ze problem, which we call the \emph{0-Extension with Steiner Nodes} problem ($\zesn$). (We note that a similar variant was mentioned in \cite{andoni2014towards}, and we provide a comparison between them in more detail in \Cref{sec: compare}.)
Upfront we are also given a function $f: \mathbb{Z}\to \mathbb{Z}$ with $f(k)\ge k$ holds for all $k\in \mathbb{Z}$, specifying the number of Steiner vertices allowed.

\vspace{-10pt}

\paragraph{0-Extension with Steiner Nodes.}
In an instance of $\zesn(f)$, the input consists of an edge-weighted graph $G=(V,E,c)$, a subset $T\subseteq V$ of $k$ vertices, that we call \emph{terminals}, and a metric $D$ on terminals in $T$, which is exactly the same as \ze.
%We denote by $\dist_{\ell}(\cdot,\cdot)$ the shortest-path (in $G$) distance metric on $V$ induced by the lengths $\{\ell_e\}_{e\in E}$.
A solution to the instance $(G,T,D)$ consists of 
\begin{itemize}
\item a partition $\fset$ of $V$ with $|\fset|\le f(k)$, such that distinct terminals of $T$ belong to different sets in $\fset$; we call sets in $\fset$ \emph{clusters}, and for each vertex $u\in V$, we denote by $F(u)$ the cluster in $\fset$ that contains it;
\item a semi-metric $\delta$ on the clusters in $\fset$, such that for each pair $t,t'\in T$, $\delta(F(t),F(t'))= D(t,t')$.
\end{itemize}
We define the \emph{cost} of a solution $(\fset,\delta)$ as $\vol(\fset,\delta)=\sum_{(u,v)\in E}c(u,v)\cdot\delta(F(u),F(v))$, and its \emph{size} as $|\fset|$.
The goal is to compute a solution $(\fset,\delta)$ with size at most $f(k)$ and minimum cost. 

The difference between $\zesn(f)$ and \ze is that, instead of enforcing every vertex to be mapped to a terminal, in $\zesn(f)$ we allow vertices to be mapped $(f(k)-k)$ non-terminals (or Steiner nodes), which are the clusters in $\fset$ that do not contain terminals. 
We are also allowed to manipulate the distances between these non-terminals, conditioned on not destroying the induced metric $D$ on terminals. Clearly, when $f(k)=k$, the $\zesn(f)$ problem degenerates to the \ze problem.
%This additional freedom makes it possible to get around the $\Omega(\sqrt{\log k})$ integrality gap

%We still require that a metric be defined on these clusters, such that the induced metric on terminal clusters (clusters in $\fset$ that contain terminals) is exactly $D$. 

It is easy to see that (\textsf{LP-Metric}) is still an LP relaxation for $\zesn(f)$, as each solution $(\fset,\delta)$ to $\zesn$ naturally corresponds to a semi-metric $\delta'$ on $V$ (where we can set $\delta'(u,u')=\delta(F(u),F(u'))$ for all pairs $u,u'\in V$). 
Denote by $\textsf{IG}_f(k)$ the worst integrality gap for (\textsf{LP-Metric}) to any $\zesn(f)$ instance with at most $k$ terminals.
In fact, similar to the connection between the integrality gap of (\textsf{LP-Metric}) and the quality achievable by flow sparsifiers \cite{moitra2009approximation,leighton2010extensions},
%following from the results in \cite{andoni2014towards}\footnote{This follows from LP(2) and Proposition 4.2 of \cite{andoni2014towards}. See \Cref{sec: compare} for a detailed discussion.}, 
it was recently shown by Chen and Tan \cite{chen20241+} that the value of $\textsf{IG}_f(k)$ is also closely related to the quality achievable by flow sparsifiers with Steiner nodes.
Specifically, for any function $f$, every graph $G$ with $k$ terminals has a quality-$\big((1+\eps)\cdot \textsf{IG}_f(k)\big)$ flow sparsifier with size bounded by $(f(k))^{(\log k/\eps)^{k^2}}$. 
This means that any positive answer to question Q2 (by proving that $\textsf{IG}_f(k)=o(\log k/\log\log k)$ for some $f$) also gives a positive answer to question Q1.

This makes it tempting to study the $\zesn$ problem. Specifically, can we prove any better-than-$O(\log k/\log\log k)$ upper bound for $\textsf{IG}_f(k)$, for any function $f$?
To the best of our knowledge, no such bound is known for any $f$, leaving the problem wide open. 
%Therefore, it is natural to ask if the additional freedom in solutions can significantly improve the old $\Omega(\sqrt{\log k})$ integrality gap. Specifically, for some small function $f$ such as $f(k)=O(k)$, can we hope to improve the integrality gap of (\textsf{LP-Metric}) to $\zesn(f)$ to $O(1)$?
In fact, no non-trivial lower bound on $\textsf{IG}_f(k)$ is known for even very small function like $f(k)=O(k)$.

\subsection{Our Results}

In this paper, we make a first step in investigating the value of  $\textsf{IG}_f(k)$, by giving a superconstant lower bound on $\textsf{IG}_f(k)$ for near-linear functions $f$.
Our main result is summarized in the following theorem.
%\Cref{thm: 0-ext with Steiner lower}

\begin{theorem}
\label{thm: 0-ext with Steiner lower}
For any $0<\eps<1$ and any size function $f:\mathbb{Z}^+\to \mathbb{Z}^+$ with $f(k)=O(k\log^{1-\eps} k)$, 
the integrality gap of the LP-relaxation (\textnormal{\textsf{LP-Metric}}) is $\textnormal{\textsf{IG}}_f(k)=\Omega(\eps \log\log k)$.
%for every sufficiently large $k$, there exists an instance $(G,T,D)$ of $\zesn$ with $|T|=k$ and size bound $f(k)$, such that the integrality gap of  (\textnormal{\textsf{LP-Metric}}) is $\Omega(\eps \log\log k)$.
\end{theorem}

%To the best of our knowledge, there is no known upper bound better than $O(\log k/\log\log k)$ for any function $f$.

We remark that our lower bound for the integrality gap of $\zesn(f)$ does not imply a size lower bound for $O(\log\log k)$-quality flow sparsifiers. However, if the same lower bound can be proved for a slightly generalized version of $\zesn(f)$, that was proposed in \cite{andoni2014towards} and analyzed in \cite{chen20241+}, then it will imply an $\Omega(k\log^{1-\eps} k)$ size lower bound for $O(\log\log k)$-quality flow sparsifiers. We provide a detailed discussion in \Cref{sec: compare}.

\subsection{Technical Overview}

We now discuss some high-level ideas in the proof of  \Cref{thm: 0-ext with Steiner lower}. Given any $k$, we will construct an unweighted graph $G$ on $n$ vertices (where $n\approx k\log k$) and $k$ terminals, and show that any solution of the $\zesn$ problem with size $O(k\log^{1-\eps} k)$ has cost lower bounded by $\Omega(\log\log k)$ times the number of edges in $G$.

Our hard instance is a constant degree expander (with an arbitrary set of its $k$ vertices as terminals). There are two main reasons to choose such a graph. First, in the previous work \cite{calinescu2005approximation} for proving the $\Omega(\sqrt{\log k})$ integrality gap lower bound for the $0$-Extension problem, an graph called ``expander with tails'' has been used. Though the tails in their construction appear useless for our purpose, as we allow a super-linear number of Steiner vertices which easily accomodate a single-edge tail for each terminal, the expander graph turns out to be still the core structure that is hard to compress.
Second, in another previous work \cite{archer2004approximate} it was shown that $0$-Extension problem on minor-free graphs has integrality gap $O(1)$, so our hard example has to contain large cliques as minor, which makes expanders, the graphs with richest structure,  a favorable choice. For technical reasons, we need some additional properties like Hamiltonian and high-girth. See \Cref{subsec: instance} for more details.

Next we want to lower bound the cost of an $\zesn$ solution of size $O(k\log^{1-\eps} k)$. This is done in two steps. We first consider a special type of solutions that really ``comes from the graph''. Recall that a solution consists of a partition $\fset$ of $V(G)$ into clusters and a metric $\delta$ on clusters in $\fset$. Specifically, in this special type of solutions, we require that each cluster $F\in \fset$ corresponds to a distinct vertex $v_F$ (called its \emph{center}) in $G$, and for all pairs $F,F'$ the metric $\delta(F,F')$ coincides with the shortest-path distance between $v_F, v_{F'}$ in $G$.
In the first step, we show that all special solutions have cost $\Omega(n\cdot \eps \log\log k)$. Intuitively, as the graph is a constant degree expander, every center $v_F$ is within distance $(\eps/100)\cdot \log\log k$ to at most $(\log k)^{\eps/10}$ other centers, but its cluster $F$ contains $(\log k)^{\eps}$ vertices on average and so it has $\Omega((\log k)^{\eps})$ inter-cluster edges. Therefore, most edges will have new length at least $(\eps/100)\cdot \log\log k$ (in $\delta$), making the total new edge length $\Omega(n \cdot \eps \log\log k)$ (as the number of inter-cluster edges in a balanced expander partition is $\Omega(n)$). Careful calculations are needed to turn these informal arguments into a rigorous proof. See \Cref{subsec: canonical} for more details.

In the second step, we show that the general solutions can actually be reduced to the special type of solutions considered in the first step, losing only an $O(1)$ factor in its cost. In fact, it has been recently shown in \cite{chen20241+} that, to analyze the cost of any $\zesn$ instance, it suffices to consider $\zesn$ solutions whose metric $\delta$ is embeddable into a geodesic structure of the terminal-induced shortest path distance metric called \emph{tight span}. Our main contribution here, on a conceptual level, is showing that \emph{for a graph with high girth, its tight span structure locally coincides with the graph structure itself}. In this sense, we can compare any solution to some special solution considered in Step 1. When the picture is lifted to the global level, in such a comparison it turns out that we will loss a factor approximately the \emph{diameter-girth ratio}, which we can manage to get $O(1)$ with an additional short-cycle-removing step in the construction of the expander. We believe this diameter-girth ratio quantifies how ``local'' a graph structure is, and should be of independent interest to other graph problems on shortest-path distances.
To carry out the technical steps, we employ a notion called \emph{continuazation} of a graph recently studied in \cite{chen2023towards}.
See \Cref{subsec: mapping} for more details.

\section{Preliminaries}

By default, all logarithms are to the base of $2$. 
%All graphs in this paper are undirected, and are allowed to have parallel edges but not self-loops.

Let $G=(V,E,\ell)$ be an edge-weighted graph, where each edge $e\in E$ has weight (or \emph{length}) $\ell_e$. 
For a vertex $v\in V$, we denote by $\deg_G(v)$ the degree of $v$ in $G$.
For each pair $S,T\subseteq V$ of disjoint subsets, we denote by $E_{G}(S,T)$ the set of edges in $G$ with one endpoint in $S$ the other endpoint in $T$.
For a pair $v,v'$ of vertices in $G$, we denote by $\dist_{G}(v,v')$ (or $\dist_{\ell}(v,v')$) the shortest-path distance between $v$ and $v'$ in $G$.
We define the \emph{diameter} of $G$ as $\diam(G)=\max_{v,v'\in V}\set{\dist_G(v,v')}$, and we define the \emph{girth} of $G$, denoted by $\gir(G)$, as the minimum weight of any cycle in $G$.
We may omit the subscript $G$ in the above notations when the graph is clear from the context.

%\paragraph{Graph expansion and a distribution of random graphs.}
%\emph{cost} of $S$ is $\vol_{G}(S)=\sum_{u\in S}\deg_{G}(u)$ and we write $\vol(G)=\vol_{G}(V)$. 
Given a graph $G$, its \emph{conductance} 
%of the cut is defined as
is defined as 
\[\Phi(G)=\min_{S\subseteq V, S \ne\emptyset, S\ne V}
\set{
\frac{|E_{G}(S,V\setminus S)|}
{\min\set{\sum_{v\in S}\deg_G(v),\sum_{v\notin S}\deg_G(v)}}
	}.\]
We say that $G$ is a \emph{$\phi$-expander} iff $\Phi(G)\ge\phi$. 

\iffalse
\paragraph{Chernoff Bound.}
We will use the following form of Chernoff bound (see e.g.\cite{dubhashi2009concentration}).
\begin{lemma}[Chernoff Bound] \label{prop:chernoff}
	Let $X_1,\dots X_n$ be independent random variables taking values in $\{0,1\}$. Let $X$ denote their sum and let $\mu=\ex{X}$ denote the sum's expected value. Then for any $\delta>0$, 
	$$
	\pr{X > (1+\delta) \mu} < \left( \frac{e^{\delta}}{(1+\delta)^{1+\delta}}\right)^{\mu}.
	$$
\end{lemma}
\fi

\section{Proof of \Cref{thm: 0-ext with Steiner lower}}

In this section we prove \Cref{thm: 0-ext with Steiner lower}. We begin by describing the hard instance in \Cref{subsec: instance}, which is essentially a high-girth expander with a subset of vertices designated as terminals. Then in \Cref{subsec: canonical} we show that a special type of solutions may not have small cost. Finally, in \Cref{subsec: mapping} we generalize the arguments in \Cref{subsec: canonical} to analyze an arbitrary solution, completing the proof of \Cref{thm: 0-ext with Steiner lower}. Some technical details in \Cref{subsec: mapping} are deferred to \Cref{subsec: diam-girth}.

%any small-sized solution to such an instance must incur a large stretch.

\subsection{The Hard Instance}
\label{subsec: instance}

Let $k$ be a sufficiently large integer.
Let $n>k$ be an integer such that $k=\ceil{\frac{n\log \log n}{\log n}}$.
Let $V$ be a set of $n$ vertices.
Let $\Sigma$ be the set of all permutations on $V$. For a permutation $\sigma\in \Sigma$, we define its corresponding edge set $E_{\sigma}=\set{(v,\sigma(v))\mid v\in V}$.

We now define the hard instance $(G,T,D)$. 
Graph $G$ is constructed in two steps. In the first step, we construct an auxiliary graph $G'$.
Its vertex set is $V$. Its edge set is obtained as follows. We sample three permutations $\sigma_1,\sigma_2,\sigma_3$ uniformly at random from $\Sigma$, and then let $E(G')=E_{\sigma_1}\cup E_{\sigma_2}\cup E_{\sigma_3}$.
In the second step, we remove all short cycles in $G'$ to obtain $G$.
%, in a similar way as \cite{erdos1959graph}.
%
Specifically, we first compute a breath-first-search tree $\tau$ starting from an arbitrary vertex of $G'$.
We then iteratively modify $G'$ as follows.
While $G'$ contains a cycle $C$ of length at most $(\log n)/100$, we find an edge of $C\setminus \tau$ (note that such an edge must exist, as $\tau$ is a tree), and remove it from $G'$. We continue until $G'$ no longer contains cycles of length at most $(\log n)/100$. We denote by $G$ the resulting graph.
The terminal set $T$ is an arbitrary subset of $V$ with size $k$. For each edge $e\in E(G)$, its weight $c(e)$ is defined to be $1$, and its \emph{length} $\ell_e$ is also defined to be $1$. The metric $D$ on the set $T$ of terminals is simply defined to be the shortest-path distance (in $G$) metric on $T$ induced by edge length $\set{\ell_e}_{e\in E(G)}$.

%Given a random permutation on vertices $\Gamma$, we define the graph $G(\Gamma)$ as a collection of cycles such that we add an edge between $v$ and $\Gamma(v)$ for each $v$. We construct a $6$-regular random graph $G$ by randomly choosing $3$ permutations $\Gamma$, $\Gamma_1$ and $\Gamma_2$, and let $G=G(\Gamma) \cup G(\Gamma_1) \cup G(\Gamma_2)$.

We next show some basic properties of the graphs $G'$ and $G$. We start with the following observation.

\begin{observation}
$G'$ is a $6$-regular graph, so $|E(G)|\le |E(G')|\le 3n$.
\end{observation}

\begin{observation}
\label{obs: girth}
$\gir(G)\ge (\log n)/100$.
\end{observation}

\begin{observation}
The probability that $|E(G')\setminus E(G)|\ge n^{0.3}$ is at most $O(n^{-0.2})$.
\end{observation}
\begin{proof}
%Recall that $\sigma_1,\sigma_2,\sigma_3$ are random permutations on $V$. We alternatively think of a random permutation being obtained as follows. We start from an arbitrary vertex of
    Let $v_1,\ldots,v_L$ be a sequence of $L\le (\log n)/100$ distinct vertices of $V$. We now show that the probability that the cycle $(v_1,\ldots,v_L,v_1)$ exists in $E(G')$ is at most $\big(6/(n-L)\big)^{L}$. Indeed, to realize the cycle edge $(v_i,v_{i+1}$, for some $\ell \in \{1,2,3\}$, $\sigma_{\ell}(v_i)=v_{i+1}$ or $\sigma_{\ell}(v_{i+1}=v_i$, where for convenience we say $v_{L+1} = v_1$. There are $6$ possible events. In order to form the cycle, we need to form $L$ edges, and each edge has $6$ possible events, which means there are at most $6^L$ ways to form the cycle in total. Consider any possible way, we have $\ell_1 \dots, \ell_{L} \in \{1,2,3\}$ and $j_1, \dots, j_L \in \{0,1\}$ such that for any index $1 \le i \le L$, we have $\sigma_{\ell_i}(v_{i+j_i})=v_{i+1-j_i}$. Let $\mathcal{E}_i$ denote this event, we have $\Pr[\mathcal{E}_i | \mathcal{E}_1 , \dots, \mathcal{E}_{i-1}] \le 1/(n-i)$. Thus the probability that all events $\mathcal{E}_i$ happen is at most $1/(n-L)^L$. Applying union bound on all the ways to form the cycle, the probability that the cycle exists in $E(G')$ is at most $\big(6/(n-L)\big)^L$.

    %\znote{todo}

Therefore, the expected number of cycles in $G'$ with length at most $(\log n)/100$ is at most 
\[\sum_{3\le L\le (\log n)/100} \frac{n(n-1)\cdots (n-L+1)}{2L\cdot \big(\frac{n-L}{6}\big)^{L}}\le \sum_{3\le L\le (\log n)/100} \frac{6^L}{2L}\cdot \bigg(1+\frac{L}{n-L}\bigg)^{L}\le \sum_{3\le L\le (\log n)/100} 6^L\le n^{0.1}. \]
Note that we delete at most one edge per each cycle, so $|E(G')\setminus E(G)|$ is less than the number of cycles in $G'$ with length at most $(\log n)/100$. Therefore, from Markov Bound, the observation follows.
\end{proof}

%It is well known that with high probability, $G$ is an expander.

We use the following previous results on the conductance and the Hamiltonicity of $G'$.

\begin{lemma} [\cite{puder2015expansion}] \label{lem:expansion}
    With probability $1-o(1)$, $\Phi(G') = \Omega(1)$.
\end{lemma}

\begin{corollary}
\label{cor: diam}
With probability $1-o(1)$, the diameter of $G$ is at most $O(\log n)$.
\end{corollary}
\begin{proof}
From the construction of $G$, $G$ contains a BFS tree of $G'$, so the diameter of $G$ is at most twice the diameter of $G'$. Therefore, it suffices to show that, if $\Phi(G') \ge \Omega(1)$, then the diameter of graph $G'$ is at most $O(\log n)$, which we do next.

Let $v$ be an arbitrary vertex of $G'$. For each integer $t$, we define the set $B_t=\set{v'\mid \dist(v,v')\le t}$, and $\alpha_t=\sum_{v': \dist_{G'}(v,v')\le t}\deg(v')$, namely  the sum of degrees of all vertices in $B_t$. 

Denote $t^*=\max\set{t\mid \alpha_t\le |E(G')|}$.
Note that, for each $1\le t\le t^*$, as $\Phi(G') \ge \Omega(1)$,
$|E(B_t,V\setminus B_t)|\ge \Omega(\alpha_t)$. Therefore, \[\alpha_{t+1}\ge \alpha_t+\sum_{v'\in B_{t+1}\setminus B_{t}}\deg(v)\ge \alpha_t+ |E(B_t,V\setminus B_t)|\ge \alpha_t\cdot (1+\Omega(1)).\]
It follows that $t^*\le O(\log n)$. Therefore, for any pair $v,v'\in V$, the set of vertices that are at distance at most $t^*+1$ from $v$ must intersect the set of vertices that are at distance at most $t^*+1$ from $v'$, as otherwise the sum of degrees in all vertices in these two sets is greater than $2|E(G')|$, a contradiction. Consequently, the diameter of $G'$ is at most $2t^*+2\le O(\log n)$.
\end{proof}

%If we only look the graph defined by $\Gamma_1$ and $\Gamma_2$, the graph contains a Hamilton cycle.

\begin{lemma} [\cite{frieze2001hamilton}] \label{lem:cycle}
With probability $1-o(1)$, the subgraph of $G'$ induced by edges of $E_{\sigma_1}\cup E_{\sigma_2}$ is Hamiltonian. 
\end{lemma}

%Let $T$ be an arbitrary set of terminals with size $k=n/\sqrt{\log n}$. We prove that even if we allow adding $k$ more centers, the $0$-extension of $G$ still cost $\Omega(n \log \log n)$. 

Now if we consider the semi-metric LP relaxation (\textsf{LP-Metric}) of this instance $(G,T,D)$, then clearly the graph itself gives a solution $\delta$ to (\textsf{LP-Metric}). Specifically, $\delta(u,u')=\dist_{\ell}(u,u')$, where $\dist_{\ell}(\cdot,\cdot)$ the shortest-path (in $G$) distance metric on $V$ induced by the lengths $\{\ell_e\}_{e\in E(G)}$. Such a solution has cost $|E(G)|=O(n)$ (as all edges have weight $c(e)=1$). Therefore, in order to prove \Cref{thm: 0-ext with Steiner lower}, it suffices to show that any solution $(\fset,\delta)$ with size $O(k\log^{1-\eps}k)$ has cost at least $\Omega(\eps n\log\log n)=\Omega(\eps n\log\log k)$.

Observe that, the graph $G$ constructed above is essentially a bounded-degree high-girth expander, which is similar to the hard instance used in \cite{leighton2010extensions} for proving the $\Omega(\log\log k)$ quality lower bound for flow vertex sparsifier (without Steiner nodes). However, our proof in the following subsections takes a completely different approach from the approach in \cite{leighton2010extensions}.

\subsection{Proof of \Cref{thm: 0-ext with Steiner lower} for a Special Class of Solutions}
\label{subsec: canonical}

In this subsection, we prove the cost lower bound for a special type of solutions to the $\zesn(f)$ instance $(G,T,D)$. Specifically, we only consider the solutions $(\fset, \delta)$ such that 
\begin{itemize}
\item each cluster $F\in \fset$ corresponds to a distinct vertex of $V$, and for each terminal $t\in T$, the unique cluster $F\in \fset$ that contains $t$ corresponds to $t$; and
\item for each pair $F,F'$ of clusters in $\fset$, if we denote by $v$ ($v'$, resp.) the vertex that cluster $F$ ($F'$, resp.) corresponds to, then $\delta(F,F')=\dist_{G}(v,v')$.
\end{itemize}
We call such solutions \emph{canonical solutions}.
In this subsection, we show that, with high probability, any canonical solution of size $o(n/\log^{\eps} n)$ has cost $\Omega(\eps n\log \log n)$. 
%As all edges in $G$ have unit length and $|E(G)|=O(n)$, such a solution incurs a stretch of $\Omega(\eps \log \log n)=\Omega(\eps \log\log k)$.

Consider now a canonical solution $(\fset,\delta)$ to the instance. For each cluster $F\in \fset$, we denote by $v(F)$ the vertex in $V$ that it corresponds to, and we call $v(F)$ the \emph{center} of $F$ (note however that $v(F)$ does not necessarily lie in $F$). We say that $F\in \fset$ is \emph{large} iff $|F|\ge n^{0.1}$, otherwise we say it is \emph{small}. 
We distinguish between the following cases, depending on the total size of large clusters.

\subsubsection*{Case 1: The total size of large clusters is at most $0.1n$}

Recall that $\vol(\fset,\delta)=\sum_{(u,u')\in E(G)}\delta(F(u),F(u'))$, where $F(u)$ ($F(u')$, resp.) is the unique cluster in $\fset$ that contains $u$ ($u'$, resp.). We call $\delta(F(u),F(u'))$ the \emph{contribution} of edge $(u,u')$ to the cost $\vol(\fset,\delta)$.
As the solution $(\fset,\delta)$ is canonical, \[\delta(F(u),F(u'))=\dist_G(v(F(u)),v(F(u'))\ge \dist_{G'}(v(F(u)),v(F(u')),\] 
as $G$ is obtained from $G'$ by only deleting edges.
%We also say that the center $v(F_u)$ \emph{represents} $u$. So the contribution of an edge is at least the distance (in $G'$) between the two centers that represent its endpoints.
%We say that an edge in $G$ is \emph{stretched}, iff its contribution is at least $\log \log n/30$.
We say that a pair $F,F'$ of clusters are \emph{friends} (denoted as $F\sim F'$), iff $\dist_{G'}(v(F),v(F'))\le \eps \log \log n/30$.
We say that an edge $(u,u')$ is \emph{unfriendly}, iff the pair of clusters that contain $u$ and $u'$ are not friends.
Therefore, in order to show $\vol(\fset,\delta)=\Omega(\eps n\log\log n)$, it suffices to show that there are $\Omega(n)$ unfriendly edges in $G'$. This is since graph $G$ is obtained from $G'$ by deleting at most $n^{0.3}$ edges, so there are $\Omega(n)$ edges contributing at least $\eps \log\log n$ each to $\vol(\fset,\delta)$.
Note that, as $G$ is a $6$-regular graph, each cluster is a friend to at most $6^{\eps \log \log n/30}<\log^{\eps/10} n$ clusters in $\fset$.
%; and if a pair $F,F'$ of clusters are not friends, then every edge in $E_G(F,F')$ is a stretched edge.

The following lemma completes the proof in this case.

%In this section, we consider the case when the total size of large cluster is at most $0.1n$. Given the clusters, we say an edge in $G$ is long if the centers of the clusters its two endpoint in has distance at least $\frac{\log \log n}{30}$. In other words, if there are $\Omega(n)$ long edges, then the total cost is $\Omega(n\log \log n)$. We say two clusters are close if distance of their centers is at most $\frac{\log \log n}{30}$. In other words, if two clusters are not close, then any edge between them is a long edge. Since each vertex in garph $G$ has degree $6$, each cluster has at most $6^{\frac{\log \log n}{30}} < \log^{-10} n$ close clusters.

%To prove a lower bound of the number of long edges, we consider the following slightly different problem: given graph $G$, we partition the vertices into at most $n / \sqrt{\log n}$ sets, and we can arbitrarily set some pair of sets are friend sets. The constraint is that each set can have at most $\log^{-10} n$ friendsets. We say an edge in $G$ is bad if two endpoint are in two different sets and they are not friend of each other. Note that any $0$-extension with at most $n / \sqrt{\log n}$ clusters correspond to one of the partition and a choice friend sets, the lower bound of the number of bad edges is also a lower bound of the number of long edges. 

%The following lemma gives the lower bound of the number of bad edges, and thus also give a lower bound of the number of long edge and the cost of $0$-extension.

\begin{lemma} \label{lem:small}
With probability $1-o(1)$, the random graph $G'$ satisfies that, for any partition $\fset$ of $V$ into $|\fset|\le O(n/\log^{\eps} n)$ clusters such that $\sum_{|F|\ge n^{0.1}}|F|\le 0.1n$, and for any friendship relation on $\fset$ in which each cluster $F$ is a friend to at most $\log^{0.1 \eps} n$ other clusters, $G'$ contains at least $n/10$ unfriendly edges, i.e., $\sum_{F\not\sim F'}|E_{G'}(F,F')|\ge n/10$.
\end{lemma}

\begin{proof}
Recall that $G'$ is obtained by sampling three random permutations $\sigma_1,\sigma_2,\sigma_3$ from $\Sigma$ and taking the union of their corresponding edge sets $E_{\sigma_1}, E_{\sigma_2}, E_{\sigma_3}$. We alternatively view $G'$ as constructed in two steps. In the first step, we obtain a graph $\hat G$ by sampling two random permutations $\sigma_1,\sigma_2$ from $\Sigma$ and letting $\hat G=(V,E_{\sigma_1}\cup E_{\sigma_2})$. In the second step, we sample a third permutation $\sigma_3$ from $\Sigma$ and let $G'=(V, E(\hat G)\cup E_{\sigma_3})$.
From \Cref{lem:cycle}, with high probability, $\hat G$ contains a Hamiltonian cycle on $V$.

    For convenience, we denote by $(\fset,\sim)$ a pair of clustering $\fset$ and the friendship relation on clusters of $\fset$. We say that the pair $(\fset,\sim)$ is \emph{valid}, iff $|\fset|\le O(n/\log^{\eps} n)$, $\sum_{|F|\ge n^{0.1}}|F|\le 0.1n$, and each cluster $F$ is a friend to at most $\log^{\eps/10} n$ other clusters.

\begin{claim}
\label{clm: number of partition}
For any Hamiltonian cycle $C$ on $V$, there are at most $n^{n/4}$ valid pairs $(\fset,\sim)$ satisfying that $\sum_{F\not\sim F'}|E_{C}(F,F')| < n/10$.
\end{claim}
\begin{proof}
%Next, we prove that fix $\Gamma_1$ and $\Gamma_2$, the number of partition and choices such that $G(\Gamma_1)\cup G(\Gamma_2)$ does not have $n/10$ bad edges is small.
Denote $L=c^*n/\log^{\eps} n$,
and let $\fset=\set{F_1,\ldots,F_L}$.

First, the number of possible friendship relations on $\fset$ such that each cluster of $\fset$ is a friend to at most $\log^{0.1\eps} n$ other clusters is at most 
%We first choose the pairs of friend sets, since there are at most $\frac{n}{\sqrt{\log n}}$ sets, and each set has at most $\log^{-10} n$ friend sets, there are at most $\binom{\frac{n}{\sqrt{\log n}}}{\log^{-10}n}$ choices for each set, and there are at most
$$
\binom{L}{\log^{0.1\eps} n}^L\le 
\binom{\frac{c^*n}{\log^{\eps} n}}{\log^{0.1\eps} n} ^{\frac{c^*n}{\log^{\eps} n}} \le \bigg(\frac{c^*n}{\log^{\eps} n}\bigg)^{\frac{c^*n}{\log^{\eps} n}\cdot \log^{0.1\eps} n}<  n^{c^*n \log^{-0.05\eps} n}.
$$
%choices in total. 
Assume now that we have a fixed friendship relation $\sim$ on the clusters in $\fset$. We now count the number of clusterings $\fset$ with $\sum_{F\not\sim F'}|E_{C}(F,F')| < n/10$.
%Next, by \Cref{lem:cycle}, $G(\Gamma_1)\cup G(\Gamma_2)$ contians a cycle, we name the vertices along the cycle as $v_1, v_2, \dots, v_n$. Before we choose the set each vertex belong to, we first choose the edges on the cycle that we allow to be a bad edge. Since there are at most $n/10$ bad edges, there are at most 
Denote $C=(v_1, v_2, \ldots, v_n,v_1)$. 
First, the number of possible unfriendly edge set (which is a subset of $E(C)$ of size at most $0.1n$) is at most
$$\sum_{i=0}^{n/10} \binom{n}{i} \le n \cdot \binom{n}{n/10} < n\cdot \bigg(\frac{en}{n/10}\bigg)^{n/10} < n^{n \log^{-0.5}n}.$$

We now count the number of clusterings $\fset$ that, together with the fixed friendship relation $\sim$, realizes a specific unfriendly edge set.
We will sequentially pick, for each $i$ from $1$ to $n$, a set among $\set{F_1,\ldots,F_L}$ to add the vertex $v_i$ to.
The first vertex $v_1$ has $L$ choices. Consider now some index $1 \le i \le n-1$ and assume that we have picked sets for $v_1,\ldots,v_{i}$. If $(v_i,v_{i+1})$ is an unfriendly edge, then vertex $v_{i+1}$ has $L$ choices; if $(v_i,v_{i+1})$ is not an unfriendly edge, this means that $v_{i+1}$ must go to some cluster that is a friend of the cluster we have picked for $v_{i}$ (or $v_{i+1}$ can go to the same cluster as $v_i$), so $v_{i+1}$ has at most $\log^{0.1\eps} n + 1$ choices.
As there are no more than $0.1n$ unfriendly edges, the number of possible clusterings $\fset$ is at most
%Finally, we choose the sets for the vertices one by one. We can choose the set for $v_1$ arbitrarily, so there are $n/\sqrt{\log n}$ choices. When we choose the set for $v_{i+1}$ for any $1 \le i \le n-1$, if we allow the edge $(v_i,v_{i+1})$ be bad edge, then there are $n/\sqrt{\log n}$ choices. However, if we do not allow the edge $(v_i,v_{i+1})$ be bad edge, then there are only $\log^{-10} n + 1$ choices. Thus, the total number of choice for the partition is at most
$$
n\cdot (\log^{0.1\eps} n)^n \cdot \bigg(\frac{n}{\log^{\eps} n}\bigg)^{0.1n} < n^{n/5}.
$$
Altogether, the number of valid pairs $(\fset,\sim)$ satisfying that $\sum_{F\not\sim F'}|E_{C}(F,F')|\ge n/10$ is at most
%Thus, the total possible partition and pairs of friend sets such that does not generate $n/10$ bad edges in $G(\Gamma_1)\cup G(\Gamma_2)$ is at most 
$$
n^{c^*n \log^{-0.05\eps} n} \cdot n^{n \log^{-0.5}n} \cdot n^{n/5} < n^{n/4}.
$$
\end{proof}

%We construct the graph $G$ in two phases, in the first phase, we sample $\Gamma_1$ and $\Gamma_2$. And in the second phase we sample $\Gamma$. 

\begin{claim}
\label{clm: number of permutation}
For every valid pair $(\fset,\sim)$, the probability that the edge set $E_{\sigma_3}$ of a random permutation $\sigma_3$ contains at most $n/10$ unfriendly edges is at most $n^{-n/3}$.
\end{claim}
%The proof has two parts. We first prove that if we fix the partition and the pairs of friend sets before the secondd phase, then the probability that there are less than $n/10$ bad edges in $G(\Gamma)$ is very small. Then we prove that given $\Gamma_1$ and $\Gamma_2$, the total number of partition and choice of friend sets such that there are no more than $n/10$ bad edges in $G(\Gamma_1) \cup G(\Gamma_2)$ is small. Finally, if we take union bound on all these partition and choice, the probabiltiy of at least one of them gives no more than $n/10$ bad edges is small.
\begin{proof}
%Fix a partition and friend pairs, we prove that the probability that $G(\Gamma)$ contains less than $n/10$ bad edges is very small. 
We say that a cluster $F\in \fset$ is \emph{bad} if it does not have a friend cluster of size at least $n^{0.4}$, otherwise we say it is \emph{good}. We first prove the following observation that most vertices lie in a bad cluster.

\begin{observation}
$\sum_{F\text{:bad}}|F|\ge 0.8n$.
\end{observation}
\begin{proof}
As the pair $(\fset,\sim)$ is valid, $\sum_{|F|\ge n^{0.1}}|F|\le 0.1n$, so $\fset$ contains at most $0.1\cdot n^{0.6}$ clusters with size at least $n^{0.4}$. 
As each cluster is a friend to at most $\log^{0.1\eps} n$ other clusters, $\fset$ contains at most $\big(0.1\cdot n^{0.6}\cdot\log^{0.1\eps}n\big)$ good sets. Therefore, the total size of all good clusters is at most $0.1n + n^{0.1} \cdot 0.1 \cdot n^{0.6}\cdot\log^{0.1\eps}n < 0.2n$. The observation now follows.
\end{proof}

%Therefore, for a random vertex $u$, than with probability at least $1-n^{-0.5}$, $(u,v)$ is a bad edge.

We alternatively construct the random permutation $\sigma_3$ as follows. We arrange the vertices in $V$ into a sequence $(v_1,\ldots,v_n)$, such that each of the first half $v_1,\ldots,v_{n/2}$ lies in some bad set. Now sequentially for each $1,2,\ldots,n$, we sample a vertex $u_i$ (without replacement) from $V$ and designate it as $\sigma_3(v_i)$. It is easy to observe that the permutation $\sigma_3$ constructed in this way is a random permutation from $\Sigma$.

The following observation completes the proof of \Cref{clm: number of permutation}.

\begin{observation}
The probability that the number of unfriendly edges in $\set{(v_i,\sigma_3(v_i))\mid 1\le i\le 9n/10}$ is less than $0.1n$ is at most $n^{-n/3}$.
\end{observation}
\begin{proof}
For any $v$ in a bad cluster, the number of vertices in its friend clusters is at most $n^{0.4} \log^{0.1} n$. 
    For each $1\le i\le 9n/10$, when we pick $\sigma_3(v_i)$, we have at least $n/10$ choices from the remaining element in $V$, and as $v_i$ is in a bad set, at most $n^{0.4} \log^{0.1} n$ of them will not create an unfriendly edge. Therefore, the probability that the edge we sample is not a bad edge is at most $\frac{1}{\sqrt{n}}$. Let $X_i$ be the indicator random variable such that $X_i = 1$ if $(v_i,\sigma_3(v_i))$ is not a bad edge. By Azuma's Inequality (Chernoff Bounds on martingales, see e.g., \cite{kuszmaul2021multiplicative}), 
    $$
        \Pr \left[ \sum_{i=1}^{2n/3} X_i > 4n/5 \right] < \left( \frac{5}{4\sqrt{n}} \right)^{4n/5} < n^{-n/3}.
    $$
    Thus, with probability at least $1-n^{-n/3}$, the set $\{(v_i,\sigma_3(v_i)) | 1 \le i \le 9n/10\}$ contains at least $9n/10-4n/5 = n/10$ bad edges.
\end{proof}
%Now we construct $\Gamma$ as follows: we sample the projection $\Gamma(v)$ for vertices one by one, but we will first pick $\Gamma(v)$ for vertices in bad sets first. For any vertex $v$ in the first half of the order, $v$ is inside a bad set, which means the number of vertices in friend sets is at most $n^{0.5}/2$, so with probability at most $\frac{n^{0.5}/2}{n/2} = n^{-0.5}$, $(v,\gamma(v))$ is not a bad edge, this is regardless of how previous vertices pick their projection. Therefore, the probability that $G(\Gamma)$ does not have at least $n/10$ bad edges is at most the probability that the sum of $n/2$ random indicator random variable with expectation $n^{-0.5}$ larger than $9n/10$. By Chernoff bound (\Cref{prop:chernoff}), the probability is at most $\left(\frac{e}{\sqrt{n}}\right)^{n} < n^{-n/3}$.
\end{proof}

Combining \Cref{clm: number of partition} and \Cref{clm: number of permutation}, we get that, over the randomness in the construction of $G'$, the probability that there exists a pair $(\fset,\sim)$ in which each cluster $F$ is a friend to at most $\log^{0.1\eps} n$ other clusters, such that $G'$ contains less than $n/10$ unfriendly edges, is at most $n^{-n/3}\cdot n^{n/4}=n^{-n/12}$.
%The proof of the lemma is finished by taking union bound on all these choices.
This completes the proof of \Cref{lem:small}.
\end{proof}

%By \Cref{lem:small}, there are at least $n/10$ edges has cost at least $\log \log n/30$ in the extension, and thus the cost of extension is $\Omega(n \log \log n)$.

\subsubsection*{Case 2: The total size of large clusters is greater than $0.1n$}

We denote by $V'$ the union of all large clusters in $\fset$.
We start by proving the following claim.

\begin{claim} \label{clm:paths}
There exists a collection of $k/4$ edge-disjoint paths in $G$, such that each path connects a distinct terminal to a distinct vertex of $V'$.
\end{claim}

\begin{proof}
We construct a graph $\hat G$ as follows. We start from graph $G'$, and add two vertices $s,t$ to it. We then connect $s$ to each terminal in $T$ by an edge, and connect each vertex in $V'$ to $t$ by an edge. All edges in $\hat G$ has unit capacity. We claim that there exists a collection $\pset$ of $k/3$ edge-disjoint paths in $\hat G$, such that each path connects a distinct terminal to a distinct vertex of $V'$. Note that this implies \Cref{clm:paths}. This is because the number of edges in $E(G')\setminus E(G)$ is at most $n^{0.3}<k/12$, and each such edge is contained in at most one path in $\pset$ (since the paths in $\pset$ are edge-disjoint), so at least $k/3-k/12\ge k/4$ paths in $\pset$ are entirely contained in $G$.
We now prove the claim.
From the max-flow min-cut theorem, it suffices to show that the minimum $s$-$t$ cut in $\hat G$ contains at least $k/3$ edges.

%Consider the following flow problem: we add a source $s$ and a sink $t$ into $G$, then we add an edge between $s$ and each terminal, and add an edge between $t$ and each vertex in large clusters. All edge has capacity $1$. To prove the claim, it is sufficient to prove that the maximum $s$-$t$ flow is at least $k/3$.

Consider any $s$-$t$ cut $(S \cup \set{s},(V \setminus S) \cup \set{t})$ in $\hat G$ and denote by $E'$ the set of edges in this cut. We distinguish between the following cases.

Case 1: $|S|\le |V|/2$. Recall that $G$ is a $6$-regular graph, so $\sum_{v\in S}\deg(v)\le \sum_{v\notin S}\deg(v)$.
Then from \Cref{lem:expansion}, $|E'|\ge \sum_{v\in S}\deg(v)/2\ge \card{S}/2$. If $\card{S} \ge 2k/3$, then $|E'|\ge k/3$. If $\card{S} < 2k/3$, then at least $k/3$ terminals lie in $V\setminus S$. As there is an edge connecting $s$ to each terminal, $|E'|\ge k/3$. 

Case 2: $|S|> |V|/2$.
Via similar arguments, we can show that $|E'|\ge |V'|/3\ge 0.1n/3> k/3$.
\end{proof}

We denote by $\pset$ the collection of paths given by \Cref{clm:paths}.
We now use these paths to complete the proof.
Consider such a path $P=(u_1,\ldots, u_r)$. Denote by $F_i$ the cluster that contains $u_i$, then the contribution of $P$ to the cost $\vol(\fset,\delta)$ is
\[\sum_{(u_i,u_{i+1})\in E(P)}\delta(F_i,F_{i+1})=\sum_{1\le i\le r-1}\dist_{G}(v(F_i),v(F_{i+1}))\ge \dist_{G}(v(F_1),v(F_{r}))\ge \dist_{G'}(v(F_1),v(F_{r})).\]
(We have used the property that for every pair $v,v'\in V$, $\dist_G(v,v')\ge \dist_{G'}(v,v')$, as $G$ is obtained from $G'$ by only deleting edges.)

Recall $P$ connects a terminal to a vertex in $V'$.
Recall that each large cluster has size at least $n^{0.1}$, so there are at most $n^{0.9}$ of them. Therefore, if we denote by $V''$ the subset of vertices that large clusters corresponds to, then $|V''|\le n^{0.9}$.
For each path $P\in \pset$, we denote by $t_P$ the terminal endpoint of $P$ (that is, $u_1=v(F_1)=t_P$), and by $v''_P$ the vertex that the cluster containing $u_r$ corresponds to (that is, $v''_P=v(F_r)$), then $\sum_{(u_i,u_{i+1})\in E(P)}\delta(F_i,F_{i+1})\ge \dist_{\ell}(t_P,v''_P)$.
As the paths in $\pset$ are edge-disjoint, their contribution to $\vol(\fset,\delta)$ can be added up, i.e., 
\begin{equation}\label{eqn: cost}
\vol(\fset,\delta)\ge \sum_{P\in \pset}\dist_{\ell}(t_P,v''_P).
\end{equation}

On the one hand, as graph $G'$ is $6$-regular, for each $v''\in V''$, the number of vertices at distance at most $\log n/100$ to $v''$ is at most $6^{\log n/100}\le n^{1/30}$. Therefore, there are at most $n^{1/30}\cdot n^{0.9}=n^{14/15}$ terms on the RHS of \Cref{eqn: cost} that at most $\le \log n/100$.
On the other hand, there are at least $k/4=\Omega(\frac{n\log \log n}{\log n})$ terms on the RHS of \Cref{eqn: cost}, so at least $k/4-n^{14/15}\ge k/5$ terms as value at least $\log n/100$. Consequently, $\vol(\fset,\delta)\ge (k/5)\cdot (\log n/100)=\Omega(k \log n)=\Omega(n\log\log n)$.

%Consider the paths given by \Cref{clm:paths} in the extension, they start with a different terminal, each of them ends with a center of a large cluster, and they are still edge disjoint, which means the cost of extension is at least the total length of these paths. Since each large cluster has size at least $n^{0.1}$, there are at most $n^{0.9}$ of them. Since the graph is a degree $6$ regular graph, each center has distance less than $\log n/100$ to at most $n^{1/30}$ terminals. Thus, there are at most $n^{14/15}<k/6$ paths has length at most $\log n/100$. Thus, the total length of these paths is at least $(k/6) \cdot (\log n/100) = \Omega(n\sqrt{\log n})$.

\newcommand{\dia}{\textsf{diam}}
\newcommand{\girth}{g}
\newcommand{\con}{\textnormal{\textsf{con}}}
\newcommand{\lcon}{\ell^{\con}}
\newcommand{\Vcon}{V^{\con}}
\subsection{Completing the Proof of \Cref{thm: 0-ext with Steiner lower}}
\label{subsec: mapping}

We have shown in \Cref{subsec: canonical} all canonical solutions with size $O(k\log^{1-\eps} n)$ have cost $\Omega(\eps n\log\log n)$. In this subsection, we complete the proof of \Cref{thm: 0-ext with Steiner lower} by showing that, intuitively, an arbitrary solution $(\fset,\delta)$ to the instance $(G,T,D)$ can be ``embedded'' into a canonical solution, without increasing its cost by too much. 

We start by introducing the notion of \emph{continuization}.

\paragraph{Continuization of a graph.}
Let $G=(V,E,\ell)$ be an edge-weighted graph. Its \emph{continuization} is a metric space $(V^{\con},\ell^{\con})$, that is defined as follows. 
Each edge $(u,v)\in E$ is viewed as a continuous line segment $\con(u,v)$ of length $\ell_{(u,v)}$ connecting $u,v$, and the point set $V^{\con}$ is the union of the points on all lines $\set{\con(u,v)}_{(u,v)\in E}$. 
Specifically, for each edge $(u,v)\in E$, the line $\con(u,v)$ is defined as \[\con(u,v)=\set{(u,\alpha)\mid 0\le \alpha\le \ell_{(u,v)}}=\set{(v,\beta)\mid 0\le \beta\le \ell_{(u,v)}},\]
where $(u,\alpha)$ refers to the unique point on the line that is at distance $\alpha$ from $u$, and $(v,\beta)$ refers to the unique point on the line that is at distance $\beta$ from $v$, so $(u,\alpha)=(v,\ell_{(u,v)}-\alpha)$.

The metric $\ell^{\con}$ on $V^{\con}$ is naturally induced by the shortest-path distance metric $\dist_{\ell}(\cdot,\cdot)$ on $V$ as follows.
For a pair $p,p'$ of points in $V^{\con}$, 
\begin{itemize}
\item if $p,p'$ lie on the same line $(u,v)$, say $p=(u,\alpha)$ and $p'=(u,\alpha')$, then $\lcon(p,p')=|\alpha-\alpha'|$;
\item if $p$ lies on the line $(u,v)$ with $p=(u,\alpha)$ and $p'$ lies on the line $(u',v')$ with $p'=(u',\alpha')$, then 
\[
\begin{split}
\lcon(p,p')=\min\{
&  \dist_{\ell}(u,u')+\alpha+\alpha', \quad
	\dist_{\ell}(u,v')+\alpha+(\ell_{(u',v')}-\alpha'),\\
&	\dist_{\ell}(v,u')+(\ell_{(u,v)}-\alpha)+\alpha', \quad
	\dist_{\ell}(v,v')+(\ell_{(u,v)}-\alpha)+(\ell_{(u',v')}-\alpha')
\}.
\end{split}
\]
\end{itemize}
Clearly, every vertex $u\in V$ also belongs to $\Vcon$, and for every pair $u,u'\in V$, $\dist_{\ell}(u,u')=\lcon(u,u')$.
For a path $P$ in $G$ connecting $u$ to $u'$, it naturally corresponds to a set $P^{\con}$ of points in $\Vcon$, which is the union of all lines corresponding to edges in $E(P)$. The set $P^{\con}$ naturally inherits the metric $\lcon$ restricted on $P^{\con}$. We will also call $P^{\con}$ a path in the continuization $(\Vcon,\lcon)$.

We show that, for each graph $G$ with a set $T$ of terminals, then any other metric $w$ on a set of points containing $T$, such that $w$ restricted on $T$ is identical to $\dist_{G}$ restricted on $T$, can be ``embedded'' into the continuation of $G$, with expected stretch bounded by some structural measure that only depends on $G$. 
Specifically, we prove the following main technical lemma.

%In this section, we give a way to embed any metric into a graph $G$.

\begin{lemma} \label{lem:d-g}
Let $(G,T,\ell)$ be any instance of $\zesn$ such that $G$ is not a tree (so $\gir(G)<+\infty$), and let $(\fset,\delta)$ be any solution to it. Let $(\Vcon,\lcon)$ be the continuization of graph $G$.
Then there exists a random mapping $\phi: \fset \to \Vcon$, such that
\begin{itemize}
\item for each terminal $t\in T$, if $F$ is the (unique) cluster in $\fset$ that contains $t$, then $\phi(F)=t$; and
\item for every pair $F,F'\in \fset$,
%we can randomly construct a projection $f$ from the nodes in $\delta$ to the veritces in $V$ such that: for any pair of nodes $F_1$ and $F_2$,
$$\ex{\lcon(\phi(F),\phi(F'))} \le O\bigg(\frac{\diam(G)}{\gir(G)}\bigg)\cdot  \delta(F,F').$$
\end{itemize}
\end{lemma}

Before we prove \Cref{lem:d-g} in \Cref{subsec: diam-girth}, we provide the proof of \Cref{thm: 0-ext with Steiner lower} using it.

\emph{Proof of \Cref{thm: 0-ext with Steiner lower}.}
Consider any solution $(\fset,\delta)$ to the instance $(G,T,\ell)$ constructed in \Cref{subsec: instance} with size $|\fset|\le o(k\cdot \log^{1-\eps}k)$. 
From \Cref{obs: girth} and \Cref{cor: diam}, $\frac{\diam(G)}{\gir(G)}\le \frac{O(\log n)}{(\log n)/100}=O(1)$. From \Cref{lem:d-g}, there exists a random mapping $\phi: \fset \to \Vcon$, such that for every pair $F,F'\in \fset$,
%we can randomly construct a projection $f$ from the nodes in $\delta$ to the veritces in $V$ such that: for any pair of nodes $F_1$ and $F_2$,
$\ex{\lcon(\phi(F),\phi(F'))} \le O(1)\cdot  \delta(F,F')$. 

Fix such a mapping $\phi$, we define a canonical solution $(\fset,\hat \delta)$ based on $(\fset,\delta)$ as follows. The collection of clusters is identical to the collection $\fset$. For each $F\in \fset$, recall that $\phi(F)$ is a point in $\Vcon$. Assume the point $\phi(F)$ lies on the line $(u,v)$ and is closer to $u$ than to $v$ (i.e., $\lcon(\phi(F),u)\le\lcon(\phi(F),v)$), then we let $u$ be the vertex in $V$ that it corresponds to. For each pair $F,F'\in \fset$, with $F$ corresponding to $u_F$ and and $F'$ corresponding to $u_{F'}$, we define $\hat\delta(F,F')=\dist_{\ell}(u_F,u_{F'})$. 
As graph $G$ in the instance $(G,T,\delta)$ constructed in \Cref{subsec: instance} is an unweighted graph, it is easy to see that
\[\hat\delta(F,F')=\dist_{\ell}(u_F,u_{F'})\le \lcon(\phi(F),\phi(F'))+2.\]
As the mapping $\phi$ is random, $\hat{\delta}$ is also random, and so $\ex{\hat\delta(F,F')}\le O(1)\cdot  \delta(F,F') +2$.

From the properties of mapping $\phi$ in \Cref{lem:d-g}, we are guaranteed that such a solution $(\fset,\hat\delta)$ is a canonical solution. Moreover, from linearity of expectation, 
$$\ex{\vol(\fset,\hat\delta)}
=\ex{\sum_{(u,v)\in E}\hat\delta(F(u),F(v))}
=\sum_{(u,v)\in E}O\bigg(\delta(F(u),F(v))\bigg)+2
=O\bigg(\vol(\fset,\delta)\bigg)+O(n).
$$
Therefore, it follows that there exists a canonical solution $(\fset,\hat\delta)$, such that $\vol(\fset,\hat\delta)\le O(\vol(\fset,\delta)+n)$.
As we have shown in \Cref{subsec: canonical} that any canonical solution $(\fset,\hat\delta)$ with $|\fset|\le o(k\log^{1-\eps}k)$ satisfies that $\vol(\fset,\hat\delta)=\Omega(\eps n\log\log n)$, it follows that $\vol(\fset,\delta)=\Omega(\eps n\log\log n)$.
%As $G$ is an unweighted graph with $O(n)$ edges, the solution $(\fset,\hat\delta)$ incurs a stretch of $\Omega(\eps \log\log n)=\Omega(\eps \log\log k)$.
This implies that the integrality gap of (\textsf{LP-Metric}) is at least $\Omega(\eps \log\log n)$.
\qed

\subsection{Proof of \Cref{lem:d-g}}
\label{subsec: diam-girth}

In this subsection, we provide the proof of \Cref{lem:d-g}. We first consider the special case where $G$ is a tree, and then prove \Cref{lem:d-g} for the general case.
%Before we prove the statement in any graph, we first prove that in the case when $G$ is a tree, we can preserve the disance. 

\begin{lemma} \label{lem:embed-tree}
Let $(G,T,\ell)$ be an instance of $\zesn$ where $G$ is a tree. Let $(\fset,\delta)$ be an solution to it. Let $(\Vcon, \lcon)$ be the continuization of $G$. Then there exists a mapping $\phi: \fset \to \Vcon$, such that
\begin{itemize}
	\item for each terminal $t\in T$, if $F$ is the (unique) cluster in $\fset$ that contains $t$, then $\phi(F)=t$; and
	\item for every pair $F,F'\in \fset$,
	%we can randomly construct a projection $f$ from the nodes in $\delta$ to the veritces in $V$ such that: for any pair of nodes $F_1$ and $F_2$,
	$\lcon(\phi(F),\phi(F')) \le \delta(F,F')$.
\end{itemize}
%Given a 0-Extension instance $(G,T,\ell_e)$ where $G$ is a tree, and given a metric $\delta$ that contains all vertices in $T$, we can randomly construct a projection $f$ from the nodes in $\delta$ to the points (including vertices and points on the edge) such that: for any pair of nodes $F_1$ and $F_2$, $\ell(\phi(F_1),\phi(F_2)) \le \delta(F_1,F_2)$.
\end{lemma}

\begin{proof}
For each terminal $t\in T$, we denote by $F_t$ the cluster in $\fset$ that contains it, and set $\phi(F_t)=t$.
For each cluster $F\in \fset$ that does not contain any terminals, we define \[\nu(F)=\min\set{\frac{1}{2}\cdot \big(\delta(F,F_t)+\delta(F,F_{t'})-\delta(F_t,F_{t'})\big)\mid t,t'\in T}.\]
%As $\delta$ is a metric on $\fset$, $\nu(F)\ge 0$. 
Denote by $t_1,t_2$ the pair of terminals $(t,t')$ that minimizes $(\delta(F,F_t)+\delta(F,F_{t'})-\delta(F_t,F_{t'}))/2$.
As $G$ is a tree, there is a unique shortest path connecting $t_1$ to $t_2$ in $G$, and therefore there exists a unique point $p$ in $\Vcon$ (that lies on the $t_1$-$t_2$ shortest path in $\Vcon$) with $\lcon(p,t_1)=\delta(F,F_{t_1})-\nu(F)$ and $\lcon(p,t_2)=\delta(F,F_{t_2})-\nu(F)$. We set $\phi(F)=p$.
	
%For any node $s$ in $\delta$, let $\nu(s)$ be the maximum distance $\nu$ such that for any two terminals $t_1$ and $t_2$ in $T$, $\delta(s,t_1)+\delta(s,t_2) \ge \delta(t_1,t_2) + 2\nu$. Since $\delta$ is a metric, we have $\nu(s) \ge 0$. Moreover, since $\nu(S)$ is the maximum possible distance, there exist two terminals $t_1(s)$ and $t_2(s)$ such that $\delta(s,t_1(s))+\delta(s,t_2(s)) = \delta(t_1(s),t_2(s)) + 2\nu(s)$. Let $\phi(s)$ be the point in the tree path between $t_1(s)$ and $t_2(s)$ such that $\ell(\phi(s),t_1(s)) = \delta(s,t_1(s)) - \nu(s)$ and $\ell(\phi(s),t_2(s)) = \delta(s,t_2(s)) - \nu(s)$. 

%Before we prove the correctness, we first prove two proposition of the projection $f$.
 
\iffalse
\begin{claim} \label{clm:t-mid}
        Given three nodes $t,t_1,t_2$ on a tree such that $t$ is on the tree path between $t_1$ and $t_2$, for any node $t'$ on the tree, $t$ is either on the tree path between $t'$ and $t_1$ or the tree path between $t'$ and $t_2$.
    \end{claim}

    \begin{proof}
        Since the path between $t'$ and $t_1$ and the path between $t'$ and $t_2$ fully contain the path between $t_1$ and $t_2$. So $t$ must be on one of these two paths
    \end{proof}
\fi

\begin{claim} \label{clm:t-proj}
For every cluster $F\in \fset$ and every terminal $t\in T$, $\lcon(\phi(F),t) \le \delta(F,F_t) - \nu(F)$. 
\end{claim}

\begin{proof}
Note that the point $\phi(F)$ lies on the (unique) shortest path between a pair $t_1,t_2$ of terminals.
For $t\in \set{t_1,t_2}$, clearly the claim holds.
Consider any other terminal $t$. Clearly $\phi(F)$ lies on either the path connecting $t$ to $t_1$ or the path connecting $t$ to $t_2$. 
Assume without lose of generality that $\phi(F)$ is on the path connecting $t_1$ and $t$.
%By \Cref{clm:t-mid}, without lose of generality, suppose $\phi(F)$ is on the path between $t_1(F)$ and $t$. 
Since $G$ is a tree, $\lcon(\phi(F),t_1)+\lcon(\phi(F),t) = \lcon(t_1,t) = \dist_{\ell}(t_1,t)$. On the other hand, by definition of $\nu(F)$ and $\phi(F)$, $\delta(F,F_{t_1})+\delta(F,F_t) \ge \delta(F_{t_1},F_{t}) + 2 \cdot\nu(F)$ holds, and $\lcon(\phi(F),t_1) = \delta(F,F_{t_1}) - \nu(F)$. Therefore, $\delta(F,F_t) \ge \lcon(\phi(F),t) + \nu(F)$.
\end{proof}
    
We now show that, for every pair $F,F'\in \fset$, $\lcon(\phi(F),\phi(F')) \le \delta(F,F')$. 
Denote by $t_1,t_2$ the pair of terminals whose shortest path contains $\phi(F)$, and by $t'_1,t'_2$ the pair of terminals whose shortest path contains $\phi(F')$.
Assume without lose of generality that $\phi(F)$ is on the tree path between $\phi(F')$ and $t_1$, so $\lcon(\phi(F'),\phi(F)) + \lcon(\phi(F),{t_1}) = \lcon(\phi(F'),{t_1})$. On the other hand, from the definition of $\phi(F)$ and \Cref{clm:t-proj}, $\lcon(\phi(F),{t_1}) = \delta(F,F_{t_1}) - \nu(F)$ and $\lcon(\phi(F'),{t_1}) \le \delta(F',F_{t_1}) - \nu(F')$. Therefore, 
$$\lcon(\phi(F),\phi(F')) \le \delta(F',F_{t_1}) - \delta(F,F_{t_1}) - \nu(F') + \nu(F) \le \delta(F,F') + \nu(F) - \nu(F').$$ Similarly, $\lcon(\phi(F),\phi(F')) \le \delta(F,F') + \nu(F') - \nu(F)$. Altogether, $\lcon(\phi(F),\phi(F')) \le \delta(F,F')$.
\end{proof}

In the remainder of this subsection, we complete the proof of \Cref{lem:d-g}.
%Now we prove the general case. Let $D$ be the diameter and $\girth$ be the girth of the graph. 
Denote $g=\gir(G)$. 
Let $r$ be a real number chosen uniformly at random from the interval $[g/60,g/30]$, so $r\le g/30$ always holds.
%Let $r<\girth/30$ be a parameter whose value is to be set later. %Let $t^*$ be an arbitrarily chosen terminal in $T$. 
For each terminal $t\in T$, we denote by $F_t$ the cluster in $\fset$ that contains it, and set $\phi(F_t)=t$, so the first condition in \Cref{lem:d-g} is satisfied.

For each cluster $F\in \fset$, we define $A_F=\min\set{\delta(F,F_t)\mid t\in T}$.
%Given a point $s$, define $A_s$ as the distance between $s$ and its closest terminal. Let $r$ be a distance which is smaller than $\girth/30$ and we will decide it later, and let $t$ be an arbitrary terminal. 
We first determine the image $\phi(F)$ for all clusters $F$ with $A_F \le r$, in a similar way as \Cref{lem:embed-tree} as follows.
%\znote{to add some intuition about the change}

%For any $s$ such that $A_s \le r$, we use similar projection scheme as the tree case, but with some modification. We project $s$ as follows: 
Define \[\nu(F)=\min\set{\frac{1}{2}\cdot \bigg(\delta(F,F_t)+\delta(F,F_{t'})-\frac{\delta(F_t,F_{t'})}{2}\bigg)\mid t,t'\in T}.\]
%As $\delta$ is a metric on $\fset$, $\nu(F)\ge 0$. 
Denote by $t_1,t_2$ the pair $(t,t')$ that minimizes the above formula. We prove the following claim.
%let $\nu(F)$ be the maximum distance $\nu$ such that for any two terminals (can be the same) $t_1$ and $t_2$ in $T$, $\delta(F,t_1)+\delta(F,t_2) \ge \delta(t_1,t_2)/2 + 2\nu$. Let $t_1(S)$ and $t_2(S)$ be those two terminals that make the inequatlity equal.

\begin{claim} \label{clm:g-tree}
$\delta(F,F_{t_1}) + \delta(F,F_{t_2}) \le 4\cdot A_F \le 4r$.
\end{claim}
\begin{proof}
Let $t$ be the terminal such that $\delta(F,F_t) = A_F$. By definition of $\nu(F)$, 
$$\nu(F) \le \frac{\delta(F,F_t)+\delta(F,F_t)- \frac{1}{2}\cdot \delta(F_t,F_t)}{2} = A_F.$$ 
On the other hand, from triangle inequality, %since $\delta(F_{t_1},F_{t_2}) \le \delta(F,F_{t_1})+\delta(F,F_{t_2})$, 
$$\nu(F) = \frac{\delta(F,F_{t_1})+\delta(F,F_{t_2}) - \frac 1 2\cdot \delta(F_{t_1},F_{t_2})}{2} \ge \frac{\delta(F,F_{t_1})+\delta(F,F_{t_2})}{4}.$$ Altogether, $\delta(F,F_{t_1})+\delta(F,F_{t_2}) \le 4\cdot A_F$. 
\end{proof}

From \Cref{clm:g-tree}, $\delta(F_{t_1},F_{t_2}) \le \delta(F,F_{t_1}) + \delta(F,F_{t_2})\le 4r < \girth/3$. Therefore, there is a unique shortest path connecting $t_1$ to $t_2$ in $G$, as otherwise $G$ must contain a cycle of length at most $2\girth/3$, contradicting the fact that $\gir(G)=\girth$. %the following claim is close to \Cref{clm:t-proj}.
 
We then set $\phi(F)$ to be the point in $\Vcon$ that lies in the $t_1$-$t_2$ shortest path, such that $\lcon(\phi(F),t_1) = 2\cdot (\delta(F,F_{t_1}) - \nu(F))$ and $\lcon(\phi(F),t_2) = 2\cdot (\delta(F,F_{t_2}) - \nu(F))$. Note that, by definition of $t_1,t_2$,
\[
\lcon(\phi(F),t_1)+\lcon(\phi(F),t_2)=2\cdot \delta(F,F_{t_1}) +2\cdot \delta(F,F_{t_2}) - 4\cdot \nu(F)=\delta(F_{t_1},F_{t_2}).
\] 
%We first prove that when $A_F$ is  small, the projection is similar to the tree case.
We prove the following claim, that is similar to \Cref{clm:t-proj}.

    \begin{claim} \label{clm:g-proj}
        For every terminal $t$ with $\delta(F,F_t) \le 6r$, $\lcon(\phi(F),t) \le 2\cdot (\delta(F,F_t) - \nu(F))$.
    \end{claim} 

    \begin{proof}
From \Cref{clm:g-tree}, 
\[
\begin{split}
\delta(F_t,F_{t_1}) + \delta(F_t,F_{t_2}) + \delta(F_{t_1},F_{t_2}) & \le  \bigg(2\cdot \delta(F,F_t) + \delta(F,F_{t_1}) + \delta(F,F_{t_2})\bigg) + \delta(F_{t_1},F_{t_2}) \\ 
& \le  12 r + 4 r + 4 r + 4 r < \girth.
\end{split}
\] 
Therefore, the point $\phi(F)$ must lie on either the $t$-${t_1}$ shortest path or the $t$-${t_2}$ shortest path in $\Vcon$, as otherwise the union of $t$-${t_1}$ shortest path, $t$-${t_2}$ shortest path, and $t_1$-${t_2}$ shortest path contains a cycle of length less than $\girth$, a contradiction.  

Assume without loss of generality that $\phi(F)$ lies on the $t$-${t_1}$ shortest path, so 
$$\lcon(\phi(F),t) = \delta(F_t,F_{t_1}) - \lcon (\phi(F),{t_1}) = \delta(F_t,F_{t_1}) - 2(\delta(F,F_{t_1}) - \nu(F)).$$ By definition of $\nu(F)$, $\delta(F_t,F_{t_1}) \le 2\cdot (\delta(F,F_t)+\delta(F,F_{t_1}) - 2\cdot\nu(F))$. Therefore, $$\lcon(\phi(F),t) \le 2\bigg(\delta(F,t)+\delta(F,F_{t_1}) - 2\nu(F) - \big(\nu(F,F_{t_1}) -\nu(F)\big)\bigg) = 2(\delta(F,F_t)-\nu(F)).$$
    \end{proof}

%Like the tree case, we are ready to prove the correctness of the projection.

We now show in the next claim that the second condition in \Cref{lem:d-g} holds for pairs of clusters that are close in $\delta$.

\begin{claim} \label{clm:g-close}
For every pair $F,F'\in \fset$ with $A_{F}, A_{F'},\delta(F,F') \le r$,  $\lcon(\phi(F),\phi(F')) \le 2\cdot\delta(F,F')$. 
    \end{claim}

\begin{proof}
Since $\delta(F,F') < r$, from triangle inequality and \Cref{clm:g-tree}, 
$$\delta(F',F_{t_1})\le \delta(F',F)+\delta(F,F_{t_1})\le r+4r\le 5r.$$ 
Similarly, $\delta(F',F_{t_2})\le 5r$. 
Then from \Cref{clm:g-proj}, 
$$\lcon(\phi(F'),t_1) \le 2\cdot (\delta(F',F_{t_1}) - \nu(F')) \le 2\cdot \delta(F',F_{t_1}) \le  10r,$$ 
and symmetrically, $\lcon(\phi(F'),t_2) < 10r$. Therefore, $$\lcon(\phi(F'),t_1)+\lcon(\phi(F'),t_2)+\lcon(t_1,t_2) < 20r + 4r < \girth.$$
%So among $t_1(F),t_2(F)$ and $\phi(F')$, one vertex is on the shortest path between the other two vertices.
Therefore, the point $\phi(F)$ must lie either on the $\phi(F')$-$t_1$ shortest path or the $\phi(F')$-$t_2$ shortest path in $\Vcon$, as otherwise the union of $\phi(F')$-${t_1}$ shortest path, $\phi(F')$-${t_2}$ shortest path, and $t_1$-${t_2}$ shortest path in $\Vcon$ contains a cycle of length less than $\girth$, a contradiction. 

Assume without loss of generality that $\phi(F)$ lies on the $\phi(F')$-$t_1$ shortest path. Then
\[
\begin{split}
\lcon(\phi(F),\phi(F')) & =\lcon(\phi(F'),t_1(F))-\lcon(\phi(F),t_1(F))\\ 
& \le 2\bigg(\delta(F',t_1(F))-\nu(F')\bigg) - 2\bigg(\delta(F,t_1(F))-\nu(F)\bigg) \le 2\bigg(\delta(F,F')+\nu(F)-\nu(F')\bigg).
\end{split}
\]  
Similarly, $\lcon(\phi(F),\phi(F')) \le 2\big(\delta(F,F')+\nu(F')-\nu(F)\big)$. Altogether, $\lcon(\phi(F),\phi(F')) \le 2\delta(F,F')$.
\end{proof}

We now complete the construction of the mapping $\phi$, by specifying the images of all clusters $F\in \fset$ with $A_F> r$. Let $t^*$ be an arbitrarily chosen terminal in $T$. For all clusters $F\in \fset$ with $A_F > r$, we simply set $\phi(F) = t^*$. 
%Now we give the complete projection scheme. We first choose an arbitrary terminal $t$ and randomly choose the distance $r$ between $\girth/60$ and $\girth/30$.
%Let $r$ be a real number chosen uniformly at random from the interval $[g/60,g/30]$. We have constructed the image $\phi(F)$ for all clusters $F$ with $A_F \le r$. 
%we use the projection scheme we described before.

It remains to show that the second condition in \Cref{lem:d-g} holds for all pairs $F,F'\in \fset$. Consider a pair $F,F'$.
%For any two points $F$ and $F'$, 
Assume first that $\delta(F,F')>\girth/60$. Then $$\lcon(\phi(F),\phi(F')) \le \dia(G) \le \frac{60 \cdot \dia(G)}{\gir(G)}\cdot \delta(F,F').$$ 
Assume now that $\delta(F,F') \le \girth/60$, and without loss of generality that $A_{F} \le A_{F'}$. 
Note that, from triangle inequality, for every terminal $t\in T$, $\delta(F',F_t)\le \delta(F,F_t)-\delta(F,F')$. This implies that $A_{F'}-A_{F} \le \delta(F,F')$. Therefore, the probability that the random number $r$ takes value from the interval $[A_{F}, A_{F'}]$ is at most $\frac{\delta(F,F')}{\girth/60}$. %In this case, $\lcon(\phi(F),\phi(F')) \le \dia(G)$. 
Note that, if $r\le A_{F}$, then $\phi(F)=\phi(F')=t^*$ and $\lcon(\phi(F),\phi(F'))=0$.
%If they are both larger than $r$, then they are both projected to $t$ and thus $\lcon(\phi(F),\phi(F'))=0$. 
And  if $r\ge A_{F'}$, then from \Cref{clm:g-close}, $\lcon(\phi(F),\phi(F')) \le 2\delta(F,F')$. Altogether,
$$\ex{\lcon(\phi(F),\phi(F'))} \le \frac{60\cdot \dia(G)}{\gir(G)} \cdot\delta(F,F') + 2 \delta(F,F') \le O\bigg(\frac{\dia(G)}{\gir(G)}\bigg)\cdot \delta(F,F').$$

%    Finally, notice that in the current scheme, we are allow to project the point to not only the vertices in the graph, but also the points between two vertices. To avoid this, we can project $s$ to the closest vertices to $\phi(F)$, we call it $f'(F)$. It clear that for any two points $F$ and $F'$, $\ell(f'(F),f'(F')) \le \ell(\phi(F_1),\phi(F_2)) + 1$.

%\input{0_ext_upper}

\appendix
\section{Comparison between $\zesn$ and the variant in \cite{andoni2014towards}}
\label{sec: compare}

\paragraph{The Steiner node variant of \ze in \cite{andoni2014towards}.}
In \cite{andoni2014towards}, the following problem (referred to as $\zesn_{\textsf{AGK}}$) was proposed.

The input consists of
\begin{itemize}
\item an edge-capacitated graph $G=(V,E,c)$,  with length $\set{\ell_e}_{e\in E}$ on its edges;
\item a set $T\subseteq V$ of $k$ terminals; and
\item \underline{a \emph{demand} $\dset: T\times T\to \mathbb{R}^+$ on terminals}.
\end{itemize}
A solution consists of 
\begin{itemize}
	\item a partition $\fset$ of $V$ with $|\fset|$, such that distinct terminals of $T$ belong to different sets in $\fset$; for each vertex $u\in V$, we denote by $F(u)$ the cluster in $\fset$ that contains it;
	\item a semi-metric $\delta$ on the clusters in $\fset$, such that:\\ \underline{$\sum_{t,t'}\dset(t,t')\cdot\delta(F(t),F(t'))\ge \sum_{t,t'} \dset(t,t')\cdot\dist_{\ell}(t,t')$},\\
	where $\dist_{\ell}(\cdot,\cdot)$ is the shortest-path distance (in $G$) metric induced by edge length $\set{\ell_e}_{e\in E(G)}$.
\end{itemize}
The cost of a solution $(\fset,\delta)$ is $\vol(\fset,\delta)=\sum_{(u,v)\in E}c(u,v)\cdot\delta(F(u),F(v))$, and its \emph{size} is $|\fset|$.
%The goal is to compute a solution $(\fset,\delta)$ with size at most $f(k)$ and minimum cost. 

From (LP1) and Proposition 4.2 in \cite{andoni2014towards}, it is proved that:
\begin{proposition}
Given a graph $G$ and a subset $T$ of its terminals, and a function $f$, if for every length $\set{\ell_e}_{e\in E(G)}$ and every demand $D$, the instance $(G,T,\ell,\dset)$ of $\zesn_{\textsf{AGK}}$ has a solution $(\fset,\delta)$ with size $|\fset|\le f(k)$ and cost 
\[\sum_{(u,v)\in E}c(u,v)\cdot\delta(F(u),F(v))\le q\cdot \sum_{(u,v)\in E}c(u,v)\cdot\dist_{\ell}(u,v),\] 
then there is a quality-$(1+\eps) q$ flow sparsifier $H$ for $G$ w.r.t $T$ with $|V(H)|\le (f(k))^{(\log k/\eps)^{k^2}}$.
\end{proposition}

\paragraph{The $\zesn$ problem.}
In studying the integrality gap of (\textnormal{\textsf{LP-Metric}}), we are essentially considering the following problem.

The input consists of
\begin{itemize}
	\item an edge-capacitated graph $G=(V,E,c)$, with length $\set{\ell_e}_{e\in E}$ on its edges; and
	\item a set $T\subseteq V$ of $k$ terminals.
\end{itemize}
A solution consists of
\begin{itemize}
	\item a partition $\fset$ of $V$ with $|\fset|$, such that distinct terminals of $T$ belong to different sets in $\fset$; for each vertex $u\in V$, we denote by $F(u)$ the cluster in $\fset$ that contains it;
	\item a semi-metric $\delta$ on the clusters in $\fset$, such that \underline{for all pairs $t,t'\in T$, $\delta(F(t),F(t'))= \dist_{\ell}(t,t')$}, where $\dist_{\ell}(\cdot,\cdot)$ is the shortest-path distance (in $G$) metric induced by edge length $\set{\ell_e}_{e\in E(G)}$.
\end{itemize}
The cost and the size of a solution is defined in the same way as $\zesn_{\textsf{AGK}}$.

The difference between two problems are underlined. Specifically, in $\zesn_{\textsf{AGK}}$ it is only required that some ``average terminal distance'' does not decrease, while in our problem it is  required that all pairwise distances between terminals are preserved.
Clearly, our requirement for a solution is stronger, which implies that any valid solution to our instance is also a valid solution to the same $\zesn_{\textsf{AGK}}$ instance (with arbitrary $\dset$). Therefore, we have the following corollary.

\begin{corollary}
Given a graph $G$ and a subset $T$ of its terminals, and a function $f$, if for every length $\set{\ell_e}_{e\in E(G)}$, the instance $(G,T,\ell)$ of $\zesn$ has a solution $(\fset,\delta)$ with size $|\fset|\le f(k)$ and cost 
	\[\sum_{(u,v)\in E}c(u,v)\cdot\delta(F(u),F(v))\le q\cdot \sum_{(u,v)\in E}c(u,v)\cdot\dist_{\ell}(u,v),\] 
	then there is a quality-$(1+\eps) q$ flow sparsifier $H$ for $G$ w.r.t $T$ with $|V(H)|\le (f(k))^{(\log k/\eps)^{k^2}}$.
\end{corollary}

On the other hand, the main result of our paper is a lower bound for the $\zesn$ problem. As $\zesn$ has stronger requirement (for solutions) than $\zesn_{\textsf{AGK}}$, our lower bound does not immediately imply a lower bound for $\zesn_{\textsf{AGK}}$ or for the flow sparsifier. However, if we can show that, for some function $f$, there exists a graph $G$, a terminal set $T$ with size $k$, and a demand $\dset$ on $T$, such that any solution $(\fset,\delta)$ with size $|\fset|\le f(k)$ has cost at least
\[\sum_{(u,v)\in E}c(u,v)\cdot\delta(F(u),F(v))\ge q\cdot \sum_{(u,v)\in E}c(u,v)\cdot\dist_{\ell}(u,v),\] 
then this, from the (LP1) and the discussion in \cite{andoni2014towards}, implies that any quality-$o(q)$ contraction-based flow sparsifier for $G$ has size at least $f(k)$.

\paragraph{Acknowledgement.} We would like to thank Julia Chuzhoy for many helpful discussions. We wan to thank Arnold Filtser for pointing to us some previous works on similar problems.

\bibliographystyle{alpha}
\bibliography{REF}

\end{document}